\documentclass[a4paper]{scrartcl}
\usepackage{amssymb}
\usepackage{amsmath}
\usepackage{amsthm}
\usepackage{mathabx}
\usepackage{latexsym}
\usepackage{enumerate}
\usepackage{graphicx,color}
\usepackage{hyperref}

\usepackage{dsfont}

\usepackage{xcolor}
\usepackage{enumitem}

\usepackage[backend=bibtex,bibencoding=ascii,style=authoryear-comp,natbib=true,dashed=false,maxcitenames=2]{biblatex} 

\renewbibmacro{in:}{%
  \ifentrytype{article}{}{\printtext{\bibstring{in}\intitlepunct}}}
\DeclareFieldFormat[article]{title}{#1}
\DeclareFieldFormat[incollection]{title}{#1}
\DeclareFieldFormat[inbook]{title}{#1}
\DeclareFieldFormat[article]{title}{#1}
\DeclareFieldFormat[unpublished,thesis]{title}{\textit{#1}}
\DeclareFieldFormat[article]{volume}{\textbf{#1}}
    \DeclareFieldFormat[article]{number}{\mkbibparens{#1}}
    \renewbibmacro*{volume+number+eid}{
      \printfield{volume}
      \setunit*{\unspace}
      \printfield{number}
      \setunit*{\unspace}
      \printfield{eid}
      }
\DeclareFieldFormat{journaltitle}{\mkbibemph{#1}\isdot}

\DeclareCiteCommand{\citeyear}
  {\boolfalse{citetracker}%
   \boolfalse{pagetracker}%
   \usebibmacro{prenote}}
  {\printtext[bibhyperref]{\printfield{year}}}
  {\multicitedelim}
  {\usebibmacro{postnote}}

\usepackage[autostyle=true]{csquotes} 

\makeatletter
\newcommand*\bigcdot{\mathpalette\bigcdot@{.5}}
\newcommand*\bigcdot@[2]{\mathbin{\vcenter{\hbox{\scalebox{#2}{$\m@th#1\bullet$}}}}}
\makeatother

\numberwithin{equation}{section}

\numberwithin{figure}{section}

\theoremstyle{plain}
\newtheorem{theorem}{Theorem}[section]
\newtheorem{proposition}[theorem]{Proposition}
\newtheorem{corollary}[theorem]{Corollary}
\newtheorem{lemma}[theorem]{Lemma}

\theoremstyle{definition}
\newtheorem{definition}[theorem]{Definition}

\newtheorem{example}[theorem]{Example}
\newtheorem{remark}[theorem]{Remark}

\usepackage{authblk}

\title{Axiomatic characterization of pointwise Shapley decompositions}

\author[1]{Marcus C.~Christiansen}

\affil[1]{\footnotesize Institut f{\"u}r Mathematik, Carl von Ossietzky Universit{\"a}t Oldenburg,  Carl-von-Ossietzky-Stra{\ss}e 9--11, \mbox{26129 Oldenburg}, Germany.}

\date{\today}

\begin{document}

\maketitle


\begin{abstract}
A common problem in various applications is the additive decomposition of the output of a function with respect to its input  variables. Functions with binary arguments  can be axiomatically decomposed by the famous Shapley value. For the decomposition of functions with real arguments, a popular method is the pointwise application of the Shapley value on the domain. However, this pointwise application largely ignores the overall structure of functions. In this paper, axioms are developed which fully preserve functional structures and lead to unique decompositions for all Borel measurable functions.
\end{abstract}

Keywords: Shapley value; average sequential decomposition; profit and loss attribution; capital allocation; explaining ML models

\section{Introduction}

The study of the effects of explanatory variables on a model  output by additive decompositions  has a long tradition in various fields of research. For example, additive decompositions are used for the analysis of poverty and inequality, cf.~Fortin et al.~(2011), for profit and loss attribution in reporting, cf.~Candland \& Lotz (2014), and for capital allocation and risk allocation in banking and insurance, cf.~Guo et al.~(2021). Since recently, additive decompositions are moreover used for explaining the output of machine learning models, cf.~Merrick \& Taly (2020).
In case of binary input variables, the model output can be interpreted as a game in collaborative game theory, and a decomposition can be obtained by means of the celebrated Shapley value, which is uniquely characterized by three axioms, see Shapley (1953). In case of non-binary input variables, a popular approach in the literature is to apply the Shapley value pointwise on the set of potential input values, see for example Shorrocks (2013), Merrick \& Taly (2020), and Godin (2022). That means that the output function is disassembled into a family of separate games, where for each possible input value a separate games is defined by activation and deactivation of single arguments. Pointwise Shapley decompositions are equivalent to average sequential decompositions, see Moehle et al.~(2021). Sequential decompositions are defined by the telescoping sums that result from sequentially activating the input variables of the function one after the other. By permutating the order in which the input factors are activated and by averaging over all possible permutations, one obtains the so-called average sequential decompositions, see Junike et al.~(2023).

Pointwise application of the Shapley value on  functions with non-binary input variables disassembles the function and
ignores its general structure, so that  Shapley's axiomatic justification is  of limited scope only. The same applies to average sequential decompositions, which are equivalent to  pointwise Shapely decompositions. This paper characterizes both of the latter decompositions by axioms that fully preserve functional structures.
While Shapley needed just three axioms for uniquely decomposing games, we use nine axioms for uniquely decomposing Borel measurable functions. The contribution of our result is twofold: First, it reveals the hidden assumptions that users  implicitly accept whenever they
apply  pointwise Shapley decompositions or average sequential decompositions. Second, it puts the latter decomposition principles on a more solid theoretical basis.

Throughout the paper, we consider functions   of $d \in \mathbb{N}$ real arguments,
$$ F: \mathbb{R}^d \rightarrow \mathbb{R},$$
which represent the output of a model.  The aim is to decompose $F$  into a sum
$$ F = G_1 + \cdots + G_d$$
of functions
$$G_i: \mathbb{R}^d \rightarrow \mathbb{R}, \quad i \in \{1,\ldots,d\},$$
where $G_i$  is meant to describe the contribution of the $i$-th input variable to the total function $F$.
The additivity makes the decomposition easy to interpret. In many applications, additivity is not just  a nice feature but actually necessary, for example, if the function $F$ describes a monetary value that shall be split between different parties according to $G_1, \ldots, G_d$..

In section 2 we recall the theory of the Shapley value. Section 3 defines average sequential definitions and shows their link to pointwise Shapley decompositions. The main contribution of this paper is section 4, where decompositions of Borel measurable functions are uniquely characterized by nine axioms. In section 5 we give three application examples.


\section{The Shapley value}\label{Section:ShapleyValue}

Shapley (1953) studied the decomposition of set-functions in cooperative game theory and derived a unique decomposition principle from three basic axioms. A game is a mapping that assigns a total gain to each possible coalition of players. Let $U=\{1, \ldots, d\}$ represent the universe of potential players. Let $\mathcal{P}(U)$ denote the power set of $U$. Then each set-function
$$ v: \mathcal{P}(U) \rightarrow \mathbb{R}$$
with the property
$$v(\emptyset)=0$$
is called a game. Shapley  (1953) additionally postulates that $v$ is superadditive, but this  assumption is actually not needed for obtaining unique decompositions, so we omit it here.
With the aim to extend Shapley's decomposition approach to more general mappings later on, we equivalently transform  the set-function $v$ to a mapping on $2^U$ by encoding participation or non-participation of the players by  binary variables,
\begin{align*}
  \chi : \{ S  : S \subseteq U \}  \rightarrow   \{0,1\}^d, \quad \chi(S) = (\mathbf{1}_{1 \in S}, \ldots, \mathbf{1}_{d \in S}).
\end{align*}
Based on the bijection $\chi$, we uniquely identify each game $v$  with a function
$$F: \{0,1\}^d \rightarrow \mathbb{R},\quad  F(x) = v(\chi^{-1}(x))= v \circ \chi^{-1}(x)$$
with the property  $$F(\mathbf{0})=0,$$
where $\mathbf{0}:=(0,\ldots,0)$. The so-called Shapley value is a vector-valued mapping
$$\phi=(\phi_1, \ldots, \phi_d): \big\{  v: \mathcal{P}(U) \rightarrow \mathbb{R}\, \big| \,v(\emptyset)=0\big\} \rightarrow \mathbb{R}^d  $$
that satisfies the equation
\begin{align*}
  v(U)  = \sum_{i \in U} \phi_i(v).
\end{align*}
The $i$-th addend $\phi_i(v)$ is  meant to describe the contribution of  player $i$ to the total gain $v(U)$. Let $\mathbf{1}:= (1, \ldots, 1)$, so that we have $v(U)=F(\mathbf{1})$.
By means of the bijection $\chi$, Shapley's decomposition $\phi(v)$ of $v(U)$ can be equivalently transformed to a decomposition $ G(\mathbf{1})$ of $F(\mathbf{1}) $,
\begin{align}\label{GfromV}
  F(\mathbf{1}) =\sum_{i \in U}\phi_i ( F  \circ \chi)  = \sum_{i=1}^d G_i(\mathbf{1}) .
\end{align}
In this way, the Shapley value $\phi$ defines a decomposition mapping for functions with binary arguments,
\begin{align*}
  \varphi=(\varphi_1, \ldots, \varphi_d): \big\{ F: \{0,1\}^d \rightarrow \mathbb{R} \, \big| \, F(\mathbf{0})=0\big\} \rightarrow \mathbb{R}^d , \quad F \mapsto  \phi \circ F  \circ \chi .
\end{align*}
For any  permutation $\pi$ on $U$, let  $\pi(S):= \{\pi(i):i\in S \}$ for $S \subseteq U$.
Shapley (1953) postulates  three axioms for the mapping $\phi$.
\begin{enumerate}
  \item[(S1)]  For each  permutation $\pi$ on $U$, let
  $$\phi_{i} (v \circ \pi ) = \phi_{\pi(i)}(v),  \quad  i \in U. $$
  \item[(S2)] For each subset $N \subseteq U$ such that $v(\cdot  \cap N) = v$, let
    \begin{align}\label{AxiomA2Shapley}  \sum_{i \in N} \phi_i(v)=v(N).
    \end{align}
  \item[(S3)] For any two games $v$ and $v'$, let   $$\phi(v+v')= \phi(v) + \phi(v').$$
\end{enumerate}
%
With a slight abuse of notation,  for  any permutation $\pi$ on $U$ and each $d$-dimensional vector $x$, let
\begin{align*}
  \pi(x):= (x_{\pi(1)}, \ldots, x_{\pi(d)}).
\end{align*}
Furthermore, for each $I \subseteq U$  we define a projection mapping $p_I$ by
 $$p_I(x):=(x_1 \mathbf{1}_{1 \in I}, \ldots, x_d \mathbf{1}_{d \in I}).$$
\begin{proposition}\label{EquivToShapley}
The mapping $\phi$ satisfies the axioms (S1), (S2), (S3) if and only if the mapping $\varphi$ satisfies the following axioms:
\begin{enumerate}
\item[(T1)] Let
  $$ F(\mathbf{1})=  G_1(\mathbf{1})+ \cdots + G_d(\mathbf{1}) $$
  for $G(\mathbf{1})=\varphi (F)$.
  \item[(T2)] For  each permutation $\pi$  on $U$ and $i \in \{1, \ldots, d\}$, let
  $$ F' = F \circ \pi  \quad \Longrightarrow \quad G'_i(\mathbf{1})= G_{\pi(i)}\circ \pi (\mathbf{1})$$
      for  $G(\mathbf{1})=\varphi (F)$ and $G'(\mathbf{1})=\varphi (F')$.
  \item[(T3)] For each $i \in \{1, \ldots, d\}$, let
  $$ F = F \circ p_{U\setminus \{i\}}  \quad \Longrightarrow\quad  G_i(\mathbf{1})=0$$
  for $G(\mathbf{1})=\varphi (F)$.
  \item[(T4)] Let
    $$F''= F+ F' \quad  \Longrightarrow  \quad G''(\mathbf{1})= G(\mathbf{1})+ G'(\mathbf{1})$$
    for  $G(\mathbf{1})=\varphi (F)$, $G'(\mathbf{1})=\varphi (F')$, and $G''(\mathbf{1})=\varphi (F'')$.
\end{enumerate}
\end{proposition}

\begin{proof}
Axiom (S1) implies that
\begin{align*}
G'_i(\mathbf{1})& = \phi_{i} (F \circ \chi \circ \pi )  \\
  &=\phi_{\pi(i)}(F \circ \xi) \\
  & = G_{\pi(i)}(\pi(\mathbf{1})),
\end{align*}
which is property (T2). For $N= U$, axiom (S2) directly yields the property  (T1). In case of $F = F\circ p_{U\setminus \{i\}}$, the corresponding game $v=F \circ \chi $ satisfies  the equation $v= v( \cdot \cap (U\setminus \{i\}) )$, so axiom (S2) implies that
\begin{align*}
   G_i(\mathbf{1}) & = \phi_i(F \circ \chi)\\
   & = \sum_{j \in U} \phi_j(F \circ \chi) -\sum_{j \in U \setminus \{i\}} \phi_j(F \circ \chi)\\
   & =v( U ) - v(U\setminus \{i\}) \\
   & = 0,
\end{align*}
which is property (T3). The property (T4) is a direct consequence of axiom (S3).

For $F':= F \circ \pi$ axiom  (T2) yields the equation
\begin{align*}
\phi_{i} (v\cdot \pi ) &=  G'_i(\mathbf{1}) \\
& = G_{\pi(i)}(\mathbf{1})\\
&= \phi_{\pi(i)}(v),
\end{align*}
which is property (S1). In case of $v=v(\cdot \cap N)$, the corresponding function $F \circ \chi^{-1}$ satisfies  $F=F\circ p_{U\setminus \{i\}}$ for each $i \not\in N$, so axioms (T1) and (T3) imply that
\begin{align*}
  v(U) & = F(\mathbf{1})\\
  & = \sum_{j =1}^d G_j(\mathbf{1}) \\
  & = \sum_{j \in N} G_j(\mathbf{1}) \\
  & = \sum_{j \in N}\phi_j(v ),
\end{align*}
which is  property (S2). The property (S3) is a direct consequence of axiom (T4).
\end{proof}
For any finite set $S$, let $|S|$ denote the number of elements of $S$. The following proposition recalls  the celebrated result of Shapley (1953).
\begin{proposition}[Shapley value]\label{PropShapleyValue}
  The mapping $\phi  $ satisfies the axioms (S1) to (S3) if and only if
  \begin{align*}
    \phi_i(v) = \frac{1}{d} \sum_{S \subseteq U} \binom{d-1}{|S|-1}^{-1} \big( v(S) - v(S\setminus \{i\})\big), \quad i \in \{1, \ldots, d\},
  \end{align*}
  for each game $v$.
\end{proposition}
Proposition \ref{EquivToShapley}  and Proposition \ref{PropShapleyValue}  yield the following corollary.
\begin{corollary}\label{Cor:ShapleyValue}
The mapping $\varphi  $ satisfies the axioms (T1) to (T4) if and only if
  \begin{align}\label{Eq:ShapleyValue}
   G_i(\mathbf{1}) = \frac{1}{d} \sum_{I \subseteq \{1, \ldots,d\}} \binom{d-1}{|I|-1}^{-1} \big( (F \circ p_I)(\mathbf{1}) - (F \circ p_{I \setminus \{i\}})(\mathbf{1}) \big), \quad i \in \{1, \ldots, d\},
  \end{align}
   for  each function $F$ and $G= \varphi(F)$.
\end{corollary}

\section{Pointwise Shapley decompositions}\label{Sec:PointwiseShaply}

We extend  the domain of $F$ from binary arguments to real-valued arguments,
\begin{align*}
  F : \mathbb{R}^d \rightarrow \mathbb{R}.
\end{align*}
We still assume that
$$ F (\mathbf{0})=0.$$
The aim is to decompose $F(x)$ for each argument $x \in \mathbb{R}^d$. A popular heuristic method is to build the telescoping sum
\begin{align*}
  F(x) &= F(x)- F(\mathbf{0})\\
  & = \sum_{i=1}^d \big( F\circ p_{\{1, \ldots, i\}}(x) -F\circ  p_{\{1, \ldots, i-1\}}(x) \big)
\end{align*}
and to interpret the addends as the contributions of each argument $x_1, \ldots, x_d$ to the total value $F(x)$. This approach is commonly denoted as sequential decomposition, cf.~Junike et al.~(2023).
For any vectors $x$ and $y$, let
 \begin{align*}
   x \ast y := (x_1y_1, \ldots, x_dy_d),
 \end{align*}
and let  $e^i =(1, \ldots,1 ,0, \ldots, 0)$ denote the $d$-dimensional vector that has the entries $1$ up to the $i$-th position and zero else. Then we can represent the sequential decomposition as
\begin{align*}
  G_i(x) = F (  e^i  \ast x ) -  F ( e^{i-1}  \ast x  ), \quad i \in \{1,\ldots, d\}.
\end{align*}
An adverse property of the sequential decomposition is its  dependence on the formal numbering or labeling of the arguments $x_1, \ldots, x_d$. Let $\pi$ be any permutation on $\{1, \ldots,d \}$. Suppose that  we  renumber the arguments $x_1, \ldots, x_d$ according to permutation $\pi$, then apply the sequential decomposition, and finally reverse the renumbering. Then we obtain the $\pi$-permutated sequential decomposition
\begin{align*}
    G_i^{\pi}(x)&=   F ( \pi^{-1}( e^{\pi(i)} \ast \pi(x)  ))   - F ( \pi^{-1}( e^{\pi(i)-1} \ast \pi(x) ))\\
    & =  F (  \pi^{-1}(  e^{\pi(i)})\ast x)   - F ( \pi^{-1}(  e^{\pi(i)-1})\ast x)  , \qquad i \in \{1,\ldots, d\},
\end{align*}
where the second equation uses the fact that $\pi( x\ast y)= \pi(x) \ast \pi(y)$ for any vectors $x$ and $y$.
For each choice of $\pi$, we end up with a different decomposition $G^{\pi}$. From a theoretical perspective, there is no distinguished permutation $\pi$ that should be preferably used, so the sequential decomposition concept is ambiguous.   In order to get rid of this ambiguity, a popular solution  is to average over all permutations $\pi$,
\begin{align}\label{ASDecomposition}
 G_i^{AS}(x) &= \frac{1 }{d!} \sum_{\pi}  \Big( F ( \pi^{-1}(  e^{\pi(i)})\ast x)   - F ( \pi^{-1}(  e^{\pi(i)-1})\ast x ) \Big), \quad i \in \{1,\ldots, d\}.
\end{align}
This averaged sequential (AS) decomposition is in fact invariant with respect to any formal renumbering or relabelling of the arguments cf.~Junike at al.~(2023).
\begin{proposition}\label{EqualsASDecomp}
It holds that
\begin{align}\label{Eq:ASShapleyValue}
  G_i^{AS}=   \frac{1}{d} \sum_{I \subseteq \{1, \ldots,d\}} \binom{d-1}{|I|-1}^{-1} \big( F \circ p_I - F \circ p_{I \setminus \{i\}} \big), \quad  i \in \{1, \ldots, d\}.
  \end{align}
\end{proposition}
\begin{proof}
For each $i \in \{1, \ldots, d\}$ and permutation $\pi$, there exists a vector $c \in \{0,1\}^d$ with $c_i=1$ and such that $ \pi^{-1}(  e^{\pi(i)})=c$ and $ \pi^{-1}(  e^{\pi(i)-1})=p_{\{1,\ldots, d\}\setminus \{i\}}(c) $. In the set of all permutations $\pi$ with $\pi(i)=r$ for arbitrary but fixed  $i,r \in \{1, \ldots, d\}$, there are subsets of size $(r-1)! (d- r)!$ that keep the vectors $\pi^{-1}(  e^{\pi(i)})$ and $\pi^{-1}(  e^{\pi(i)-1})$ constant.
Therefore,
\begin{align*}
 G_i^{AS}(x)&= \frac{1 }{d!} \sum_{\pi}  \Big( F (  \pi^{-1}(  e^{\pi(i)})\ast x)   - F (  \pi^{-1}(  e^{\pi(i)-1})\ast x) \Big)\\
  &= \frac{1}{d!}\sum_{r=1}^d\sum_{c \in \{0,1\}^d \atop  \|c\|_1 = r, c_i=1 } (r-1)! (d- r)!  \Big( F ( c \star x) - F( p_{\{1,\ldots, d\}\setminus \{i\}}(c) \star x)\Big)\\
  & = \sum_{c  \in \{0,1\}^d\atop c_i=1 } \frac{(\|c\|_1-1)! (d- \|c\|_1)!}{d! } \big( F(c \star x)  - F(p_{\{1,\ldots, d\}\setminus \{i\}}(c) \star x)\big).
\end{align*}
In the last sum we can drop the condition $c_i=1$ since the addend is anyway zero if $c_i=0$, so the last term is equivalent to  \eqref{Eq:ASShapleyValue}.
\end{proof}

Let $\mathcal{F}$ denote the set of real-valued  functions on $\mathbb{R}^d$. By $\mathcal{F}_{0}$ we denote the subset of functions $F\in \mathcal{F}$ with the property $F(\mathbf{0})=0$.

\begin{definition}
 Let $\delta^{AS}: \mathcal{F}_{0} \rightarrow \mathcal{F}^d$ be defined as the mapping that assigns to each function $F \in \mathcal{F}_{0} $ the decomposition \eqref{ASDecomposition}. We call $\delta^{AS}$  the AS decomposition principle.
\end{definition}

In the special case of $x=\mathbf{1}$,  formula \eqref{Eq:ASShapleyValue} equals   the Shapley value \eqref{Cor:ShapleyValue}, so the AS decomposition principle may be seen as a generalization of the Shapley value.
However,  the Shapley value is based on axiomatic principles, whereas  the AS decomposition principle is just based on a heuristic concept. Yet, as Shorrocks (2013) explains, the AS decomposition principle can be derived from a  pointwise application of the Shapley value, so that the AS concept gets an axiomatic foundation: For each $x \in \mathbb{R}^d$, define a decomposition $G(x)$ of $F(x)$ by applying the Shapley value on the mapping
\begin{align*}
  F^x : \{0,1\}^d \rightarrow \mathbb{R}^d, \quad y \mapsto  F(x \ast y).
\end{align*}
As the following proposition shows, this pointwise construction indeed establishes the AS decomposition.
\begin{proposition}\label{PointwiseShapley}
For  each  $F \in \mathcal{F}_{0}$  and  $x \in \mathbb{R}^d$, it holds that
\begin{align*}
  \delta^{AS}( F)(x)=\varphi(F^x).
\end{align*}
\end{proposition}
\begin{proof}
  For each  $x \in \mathbb{R}^d$ and $I\subseteq \{1, \ldots, d\}$, it holds that $ F^x \circ p_I (\mathbf{1})=  F(x \ast  p_I(\mathbf{1})) = F ( p_I (x))$. By applying this fact in \eqref{Eq:ASShapleyValue}, we obtain  that $G^x_i(\mathbf{1})$ equals \eqref{Eq:ASShapleyValue} for each $i \in \{1, \ldots, d\}$. According to Proposition \ref{EqualsASDecomp}, this means that $G^x(\mathbf{1})=G^{AS}(x)$.
\end{proof}

\section{Axiomatic functional decompositions}

In the previous section we derived the AS decomposition principle by applying  the Shapley value pointwise on the domain of $F$, but this pointwise construction largely ignores the general functional structure of $F$. For each $x \in \mathbb{R}^d$, the decomposition of $F(x)=F^x(\mathbf{1})$ by the axioms (T1) to (T4) involves only the function's values on the finite subset
$\{p_I(x): I \subseteq \{1, \ldots, d\}\}\subset \mathbb{R}^d$,
and the  structure of $F$ on the remaining domain is completely ignored by (T1) to (T4).   This  chapter presents decomposition axioms  that preserve $F$ as entire function on $\mathbb{R}^d$.

We still consider functions with real-valued arguments,
\begin{align*}
  F : \mathbb{R}^d \rightarrow \mathbb{R}
\end{align*}
but we are not assuming $ F (\mathbf{0})$ to be zero anymore.  A mapping $\delta :\mathcal{F} \rightarrow \mathcal{F}^d$ that assigns to each function $F \in \mathcal{F}$ a decomposition $G=\delta(F) \in \mathcal{F}^d$
is called a decomposition principle.
\begin{definition}
Let $\delta^{\ast}: \mathcal{F} \rightarrow \mathcal{F}^d$ be defined as the decomposition principle that assigns to each function $F \in \mathcal{F} $ the decomposition
\begin{align}\label{Eq:ASShapleyValueExt}
   G_i=  \frac{1}{d} F(\mathbf{0}) +
    \frac{1}{d} \sum_{I \subseteq \{1, \ldots,d\}} \binom{d-1}{|I|-1}^{-1} \big( F \circ p_I - F \circ p_{I \setminus \{i\}} \big) ,\quad i \in \{1,\ldots, d\}.
  \end{align}
\end{definition}
The restriction of $\delta^{\ast}$ to $\mathcal{F}_0$ is equal to the AS decomposition principle,
\begin{align*}
  \delta^{\ast}|_{ \mathcal{F}_0} = \delta^{AS}.
\end{align*}

 For any mappings $h_1, \ldots, h_d: \mathbb{R} \rightarrow \mathbb{R}$, let $F(h_1, \ldots, h_d)$ denote the mapping
\begin{align*}
  x \mapsto F(h_1(x_1), \ldots, h_d(x_d)).
\end{align*}
\begin{proposition}\label{ASsatisfiesA1A9}
If a decomposition principle  $\delta$ equals $\delta^*$, then it   satisfies the following axioms:
\begin{enumerate}
\item[(A1)] Let  $$ F= G_1 + \cdots + G_d $$
for $G= \delta (F)$.
\item[(A2)]  For any permutation $\pi$, let
  $$ F' = F\circ \pi \quad \Longrightarrow \quad G'_i= G_{\pi(i)}\circ \pi  $$
  for $i \in \{1, \ldots, d\}$ and  $G= \delta (F)$, $G'= \delta (F')$.
\item[(A3)] Let
  $$ F = F \circ p_{\{1,\ldots,d\}\setminus \{i\}}  \quad \Longrightarrow\quad  G_i=G_i \circ p_{\emptyset} $$
  for $G= \delta (F)$.
\item[(A4)] Let
  $$F''= F+ F' \quad  \Longrightarrow  \quad G''= G+ G'$$
  for $G= \delta (F)$, $G'= \delta (F')$, $G''= \delta (F'')$.
  \item[(A5)] For any $\alpha \in \mathbb{R}$, $\alpha \neq 0$,  let
  $$  F'=\alpha F  \quad  \Longrightarrow  \quad G'=\alpha G $$
  for   $G= \delta (F)$, $G'= \delta (F')$.
\item[(A6)] Let
  $$ F = F \circ p_{\{1,\ldots,d\}\setminus \{i\}}  \quad \Longrightarrow\quad  G= G \circ p_{\{1,\ldots,d\}\setminus \{i\}}$$
  for   $i \in \{1, \ldots, d\}$ and  $G= \delta (F)$.
\item[(A7)] If the pointwise limit $\lim_{n \rightarrow \infty} F^n$ exists, 
    then let
    $$ F= \lim_{n \rightarrow \infty} F^n \quad \Longrightarrow \quad  G= \lim_{n \rightarrow \infty} G^n $$
    for  $G= \delta (F)$, $G^n= \delta (F^n)$, $ n \in \mathbb{N}$.
 \item[(A8)] For any sequence $(x^n)_{n \in \mathbb{N}}$ with $x^n \rightarrow x \in \mathbb{R}^d$, let  $$ \lim_{n \rightarrow \infty } F(x^n) = F(x)  \quad \Longrightarrow \quad  \lim_{n \rightarrow \infty } G(x^n) = G(x)$$
     for  $G= \delta (F)$.
 \item[(A9)] For  any homeomorphisms  $h_1, \ldots, h_d$ on $\mathbb{R}$ with  fixed point zero, let  $$ F' = F(h_1, \ldots, h_d)   \quad \Longrightarrow \quad  G' = G (h_1, \ldots, h_d)$$
      for $G= \delta (F)$, $G'= \delta (F')$.
\end{enumerate}
\end{proposition}
Before we give the proof of Proposition \ref{ASsatisfiesA1A9}, we present a useful equivalent characterization of the assumptions in axiom (A9).
\begin{lemma}\label{Lemma:ZeroPreserv}
For any mapping  $g: \mathbb{R}^d \rightarrow \mathbb{R}^d $,
the two following statements are equivalent:
\begin{enumerate}
  \item[(a)] There exist  homeomorphisms $h_1, \ldots, h_d$ on $\mathbb{R}$ with fixed point zero such that
  $$g(x)= (h_1(x_1), \ldots, h_d(x_d)), \quad x \in \mathbb{R}^d.$$
  \item[(b)] The mapping  $g $ is a homeomorphism with the property
  $$g\circ p_I= p_I \circ g, \quad I \subseteq \{1,\ldots, d\}.$$
\end{enumerate}
\end{lemma}
\begin{proof}
Statement (b) follows from (a) because of $h_i(0)=0$ for all $i$. Statment (a) follows from (b) since $g\circ p_{\{i\}}= p_{\{i\}} \circ g$  means that  $g_i(x_1, \ldots, x_d)=g_i(0 \ldots, 0, x_i,0, \ldots, 0)$, so that we can set $h_i(x_i):=g_i(0 \ldots, 0, x_i,0, \ldots, 0)$.
\end{proof}

\begin{proof}[Proof of Proposition \ref{ASsatisfiesA1A9}]
The constant part $F(\mathbf{0})$ of $F$ satisfies axiom (A1). For the remaining part $F(x) - F(\mathbf{0})$, we first note that the SU decomposition and all permutated SU decompositions  satisfy axiom (A1), since they are defined from   telescoping sum of $F(x) - F(\mathbf{0})$. The AS decomposition, which just averages all permutations of SU decompositions, must then also satisfy axiom (A1).

Given that $F' = F \circ \pi$, it holds that
\begin{align*}
   G'_i(x) & =  \frac{1}{d}F\circ \pi (\mathbf{0}) +\frac{1}{d} \sum_{I \subseteq \{1, \ldots,d\}  } \binom{d-1}{|I|-1}^{-1} \big( F \circ \pi \circ p_{I}  -  F \circ \pi \circ p_{I \setminus \{i\}}\big) \\
   & =    \frac{1}{d}F(\mathbf{0}) + \frac{1}{d} \sum_{J \subseteq \{1, \ldots,d\}  } \binom{d-1}{| J| -1}^{-1} \big(  F \circ p_{J } \circ \pi   - F \circ p_{J \setminus \{\pi(i)\}} \circ \pi \big) \\
   & = G_{\pi(i)}\circ \pi ,
\end{align*}
which verifies axiom (A2).

If $F = F \circ p_{\{1,\ldots,d\}\setminus \{i\}}$, then $F\circ p_I = F \circ p_{I \setminus \{i\}}$ for all $I \subseteq \{1, \ldots,d\}$, so that  all addends in \eqref{Eq:ASShapleyValue} are zero except for $\frac{1}{d}F(\mathbf{0})$. So $G_i$ is constant, which implies $G_i=G_i \circ p_{\emptyset}$.

For verifying axioms (A4) to (A9),  we use the fact that \eqref{Eq:ASShapleyValue} defines  $G$ directly from $F$, so that the functional properties of $F$ directly translate to analogous  properties for $G$. In particular, we use the fact that $p_{\{1,\ldots,d\}\setminus \{i\}} \circ p_I = p_I \circ p_{\{1,\ldots,d\}\setminus \{i\}}$, the continuity of   $p_I$, and Lemma \ref{Lemma:ZeroPreserv}.
\end{proof}

In axiom (A5) we excluded the case $\alpha=0$. This case is already covered by the other axioms.

\begin{lemma}\label{CaseAlphaNull}
The axioms (A1), (A2), (A6) imply the statement of axiom (A5)  for $\alpha =0 $.
\end{lemma}
\begin{proof}
In case of $F'=0$, it holds that $F' = F' \circ p_{\emptyset}$ and $F'=F' \circ \pi$, so the axioms  (A6) and (A2)  imply that
  $$ G'_i(x)=G'_i(\mathbf{1}) = G'_{\pi(i)} (\mathbf{1})=G'_{\pi(i)} (x)$$
  for all $x$ and $i$. Because of axiom (A1), that means that
  $$0 = F'(x) = d \, G'_i(x),$$
   which verifies that $G'= \mathbf{0}$.
\end{proof}
\begin{remark}Here we briefly interpret the axioms:
Axiom (A1) is the starting assumption of this paper and is added for completeness. Axiom (A2) says that  the decomposition principle shall be invariant with respect to any formal renumbering or relabelling of the arguments $x_1, \ldots, x_d$.  Axiom (A3) says that an argument $x_i$ that has no impact on the  function $F$ shall have a constant contribution function $G_i$. In the AS decomposition principle we even have $G_i=0$ in this case, so that the argument $x_i$ makes no contribution at all. We use the weaker postulate  $G_i=const$ in order to not rule out constant functions $F$.
The axioms (A4) and (A5) could be combined to a 'linearity axiom', including the case $\alpha=0$, see Lemma \ref{CaseAlphaNull}. Axiom (A6) says that an argument $x_i$ that has no impact on the function $F$ shall likewise have no impact on the decomposition. Axiom (A7) postulates continuity of the mapping $\delta$. Axiom (A8) says that a potential continuity of $F$ shall be inherited by $G$. Axiom (A9) postulates that the decomposition principle shall be invariant with respect to lossless data conversions of the arguments $x_1, \ldots, x_d$. The fixed point assumption makes sure that the data conversion does not shift the reference point  $\mathbf{0}$.
\end{remark}

Let $\mathcal{M}\subset \mathcal{F}$ denote the subset of Borel-measurable functions. By $\mathcal{M}_{0}$ we denote the subset of functions $F\in \mathcal{M}$ with the property $F(\mathbf{0})=0$.

\begin{theorem}\label{GenStaticDecompos}
If a decomposition principle $\delta:\mathcal{F} \rightarrow \mathcal{F}^d$ satisfies the axioms (A1) to (A9), then
\begin{align*}
  \delta|_{\mathcal{M}} = \delta^{\ast}|_{\mathcal{M}}.
\end{align*}
\end{theorem}
The relation $\delta^{\ast}|_{ \mathcal{F}_0} = \delta^{AS}$ immediately implies the following corollary.
\begin{corollary}\label{GenStaticDecompos2}
If a decomposition principle $\delta:\mathcal{F}_0 \rightarrow \mathcal{F}^d$ satisfies the axioms (A1) to (A9), then
\begin{align*}
  \delta|_{\mathcal{M}_0} = \delta^{AS}|_{\mathcal{M}_0}.
\end{align*}
\end{corollary}

\begin{proof}[Proof of Theorem \ref{GenStaticDecompos}]
First of all, we consider a constant function $F$. Then  we have $F = F \circ \pi$ and $F=F \circ p_I$ for any permutation $\pi$ and $I \subseteq   \{1,\ldots,d\}$, so that the axioms (A1), (A2), (A6)  imply that
\begin{align*}
  F(x)= G_1(x) + \cdots + G_d(x) =  G_1(\mathbf{0}) + \cdots + G_d(\mathbf{0}) =d\,  G_j(\mathbf{0})
\end{align*}
for each $j \in \{1,\ldots,d\}$. This verifies \eqref{Eq:ASShapleyValueExt} for constant functions.

In a second step, we consider functions of type
\begin{align}\label{FAsPolynimial0}
  F(x) = \prod_{i=1}^d (\max\{ s_i x_i,0 \})^{q_i}
\end{align}
for  $ q \in\mathbb{N}_0^d$ and  $s \in \{1,-1\}^d$. The case $ q  =\mathbf{0}$ has been already covered above, so let now $q \neq \mathbf{0}$.
At first, we just consider  $q \in \{0,1\}^d$ and $s= \mathbf{1}$.  For $q$ arbitrary but fixed, we define
 $I_1= \{ i: q_i= 1\}$ and $I_0 = \{ i: q_i=0\}$. Let $\pi$ be a permutation with the property $\pi(I_1)= I_1$. Let $y \in \mathbb{R}^d$ be a vector such that $\pi(y)=y$. From axiom (A2) we can conclude that
$G_i(y) = G_{\pi(i)}(y)$ for all $i \in \{1,\ldots, d\}$. This fact and axioms (A1) and (A3) imply that
\begin{align}\label{RepresFcTwoParts}\begin{split}
   F (y)&= G_1 (y) +  \ldots +  G_d (y) \\
   &= (d-\| q\|_1) G_i (y) +\| q\|_1 G_j (y)\\
   &= (d-\| q\|_1) G_i (\mathbf{0}) +\| q\|_1 G_j (y), \quad i \in I_0, j \in I_1.
\end{split}\end{align}
For $h_1, \ldots, h_d$ defined by $ h_i(x_i):= \beta_i x_i $ for $\beta_i \neq 0$, it holds that
$$F(h_1, \ldots, h_d) = \Big( \prod_{i \in I_1} \beta_i \Big) F,$$
so that axiom (A9) and  axiom (A5) with $\alpha =  \prod_{i\in I_1} \beta_i $  imply that
\begin{align}\label{RepresFux}
  G(\beta_1 x_1 , \ldots,  \beta_d x_d) =  \Big( \prod_{i \in I_1}\beta_i \Big) G(x) \quad \forall x.
\end{align}
For $\beta^n = (1, \ldots, 1,\varepsilon_n,1 \ldots, 1) $ with the variable $\varepsilon_n$ at the $i$-th position and $\varepsilon_n \downarrow 0$ for $n \rightarrow \infty$, equation  \eqref{RepresFux} and
axiom (A8) yield
\begin{align*}
G( p_{\{1,\ldots,d\}\setminus \{i\}}(x))
= \lim_{n \rightarrow \infty} \varepsilon_n   G( x )=0, \quad i \in I_1,\, x \in \mathbb{R}^d,
\end{align*}
since $F$ is continuous. By repeating this type of argument for each $i  \in I_1$ and by applying axiom (A6),  we conclude that
\begin{align}\label{GbeiNull}
G(x)=0 \; \forall x : x_1 x_2 \cdots x_d=0 .
\end{align}
Furthermore, because of axiom (A3) it holds that
\begin{align*}
G_j(x)= G_j(\mathbf{0}) =0, \quad \forall x, j \in I_0,
\end{align*}
so that  equation \eqref{RepresFcTwoParts} becomes
\begin{align}\label{RepresFcOnePart}
     F (y)= \| q\|_1 G_j (y)  , \quad  j  \in I_1.
\end{align}
Let $z$ be a vector such that $z_i \neq 0$  for all $i \in I_1$.  Then, for $h_1, \ldots, h_d$ defined by $h_i(x_i):=\beta_ix_i$ with  $\beta_i = z_i$ for $i \in I_1$ and $\beta_i=1$ for $i \in I_0$,  axiom (A6) and  the equations \eqref{RepresFux} and \eqref{RepresFcOnePart} yield
 \begin{align*}
     G_j (z) &=G_j(\beta_1 , \ldots, \beta_d )\\
      &= \Big( \prod_{i \in I_1} \beta_i \Big)  G_j (\mathbf{1})\\
      &=  \frac{1}{\| q\|_1}  \Big( \prod_{i \in I_1} \beta_i \Big) F(\mathbf{1})\\
      & = \frac{1}{\| q\|_1}  F(z) , \quad j \in I_1,
\end{align*}
since $(\prod_{i \in I_1} \beta_i )F(\mathbf{1}) =F(z)$. All in all, from the latter equation and \eqref{GbeiNull} we conclude that
\begin{align}\label{RepresGOnePart}
     G_j (x) = q_j \frac{1}{\| q\|_1} F(x)  \; \forall x, j.
\end{align}
For $j \in I_0$ the latter equation equals \eqref{Eq:ASShapleyValue}, since  all addends are zero. Now suppose that $j \in I_1$. Since $F( p_I(x))  - F( p_{I \setminus \{j\}}(x)) \neq 0$ only if $ I \supseteq I_1$ and since $F( p_I(x))  - F( p_{I \setminus \{j\}}(x)) = F( x) $ for all $ I \supseteq I_1$, by using  $\mathbf{1}_{I \supseteq I_1}= F (p_I(\mathbf{1}))$ we can show that
\begin{align*}
  &   \frac{1}{d} \sum_{I \subseteq \{1, \ldots,d\}  } \binom{d-1}{\| I|-1}^{-1} \big( F( p_I(x))  - F( p_{I \setminus \{i\}}(x))\big)\\
  &=  \frac{1}{d} \sum_{I \supseteq I_1  } \binom{d-1}{| I|-1}^{-1} F(x)\\
  &=  F(x) \frac{1}{d} \sum_{I } \binom{d-1}{| I|-1}^{-1} \mathbf{1}_{ I \supseteq I_1}\\
  &=  F(x) \frac{1}{\|q\|_1},
\end{align*}
where we refer to Shapley (1953, section 3) for the last equality. This verifies \eqref{Eq:ASShapleyValue} for functions  \eqref{FAsPolynimial0} with $q \in \{0,1\}^d$ and $s= \mathbf{1}$.  We can expand that result to general exponents $q \in \mathbb{N}_0$ by applying axiom (A9) for  $h_1, \ldots, h_d$ defined by $$h_i(x_i) :=  \mathrm{sign}(x_i)^{ \min \{ q_i ,1\}} |x_i|^{\max \{ q_i,1\} }.$$ Moreover, we  can expand our result to any $s \in \{1,-1\}^d$ by applying axiom (A9) with $h_i(x_i):= s_i x_i$.

In a next step we consider any function $F$ of type
  \begin{align}\label{FAsPolynimial01}
F(x)=  \prod_{i=1}^d x_i^{q_i}
\end{align}
for $q\in \mathbb{N}_0$, which can be represented as a linear combination of functions of type \eqref{FAsPolynimial0},
\begin{align*}
F(x) = \sum_{s \in \{1,-1\}^d} \Big(\prod_{i=1}^d s_i^{q_i} \Big) \Big( \prod_{i=1}^d (\max\{ s_i x_i,0 \})^{q_i}\Big),
\end{align*}
because of $$x_i^{q_i}= \bigg(\sum_{s_i \in \{1,-1\} } s_i \max\{s_ix_i,0\}\bigg)^{q_i}=\sum_{s_i \in \{1,-1\} } s_i^{q_i} (\max\{s_ix_i,0\})^{q_i}.$$ Now we apply axioms (A4) and (A5) in order to verify   \eqref{Eq:ASShapleyValue} for functions of type \eqref{FAsPolynimial01}.

In our next step let $F$ by any polynomial, i.e.~$F$ is a linear combination of functions of type \eqref{FAsPolynimial01}. Then axioms (A4) and (A5) imply  \eqref{Eq:ASShapleyValue}.

Now let $F$ be any continuous function. According to the  Stone-Weierstrass theorem, for each $\varepsilon_n>0$ and $b_n>0$ there exists a polynomial function $F^n$ such that $\sup_{x \in [-b_n,b_n]^d} |F(x)-F^n(x)| < \varepsilon_n$. For example, approximate  $F|_{[-b_n,b_n]^d}$ by Bernstein polynomials. For  sequences $\varepsilon_n \rightarrow 0$ and $b_n \rightarrow \infty$, we can construct a sequence of polynomial approximations $(F^n)_{n \in \mathbb{N}}$ that converges pointwise to $F$. For each polynomial $F^n$, the  formula \eqref{Eq:ASShapleyValue} applies, which is linear in $F^n$, so that the limit $ \lim_{n \rightarrow \infty} G^n$ exists as a pointwise limit. According to axiom (A7), this verifies \eqref{Eq:ASShapleyValue} for each continuous function $F$.

By iteratively repeating the latter step of building limits of sequences of functions and applying axiom (A7), starting from the set of continuous functions, we obtain \eqref{Eq:ASShapleyValue} for each step function and finally for each measurable function.
\end{proof}

\section{Examples}

The three examples in this section illustrate that model outputs are typically  Borel measurable functions, so that the axioms (A1) to (A9) imply unique decompositions,  see Theorem \ref{GenStaticDecompos}.

\begin{example}\label{Exp1}
Consider the gains and losses between time $0$ and time $1$ of a stock in foreign currency,
\begin{align*}
  F(x_1,x_2) = (x_1+s_0)(x_2+c_0)- s_0c_0,
\end{align*}
where $s_0$ and $s_1=s_0+x_1$ are the stock values in foreign currency at times $0$ and $1$, and $c_0$ and $c_1=c_0+x_2$ are the currency exchange factors into home currency. The gains and losses  shall be explained from the variables $x_1=s_1-s_0$ and $x_2=c_1-c_0$.
The function $F$
is Borel measurable and has the property $F(0,0)=0$, so Corollary \ref{GenStaticDecompos2} suggests for $F(x_1,x_2)$ the decomposition
\begin{align*}
  G_1(x_1,x_2) = \frac{x_1x_2}{2}+  x_1c_0, \quad G_2(x_1,x_2) = \frac{x_1x_2}{2}+  x_2s_0.
\end{align*}
\end{example}
The decomposition principle $\delta^{\ast}$ expands the AS decomposition principle to functions $F$ that are not necessarily zero at zero. This is relevant in the following example.
\begin{example}\label{Exp2}
Let $x_1, \ldots, x_d$ be the electricity meter readings of  $d \in \mathbb{N}$ individuals in a shared housing community who share a single utility contract.  The total utility bill  is given by an increasing  cost function $f$ of the total electricity consumption,
$$ F(x_1, \ldots, x_d) = f( x_1 + \ldots +x_d).$$
The mapping  $f$ may be non-linear due to volume discounts.
The total bill shall be split among the $d$ individuals according to their electricity meter readings $x_1, \ldots, x_d$.
Since the cost function $f$ was assumed to be increasing, the function $F$ is Borel measurable, so Theorem \ref{GenStaticDecompos} suggests to decompose $F$ by $\delta^*(F)$. The first addend
$$\frac{1}{d}F(\mathbf{0})= \frac{1}{d}f(0)$$
in \eqref{Eq:ASShapleyValueExt} describes consumption-independent fixed costs, which are  evenly split between the $d$ individuals. The remaining part of \eqref{Eq:ASShapleyValueExt} equals the AS decomposition $\delta^{AS}(F)$ and addresses consumption-dependent costs.
\end{example}

\begin{example}\label{Exp3}
Profits and losses that emerge in an insurer's balance sheet between two reporting dates  stem from various risk sources. International reporting standards as well as insurance regulation require a change analysis of the insurer's balance sheet with the aim to  identify and quantify the sources of the observed profits and losses. For example, let $C$  be a random variable that describes  future  insurance claims that are evaluated by  the risk measure Value  at Risk to the level $99.5 \%$.
  Suppose that the insurer's risk model comprises $d$ risk factors, and let the real-valued random variables  $X_1, \ldots, X_d$  describe changes in these risk factors from one reporting date to the next.
    So the evaluation of claim $C$ changes between the reporting dates by
\begin{align*}
  F(x_1, \ldots, x_d) =  \mathrm{VaR}_{0.995}\big[C\big| X_1=x_1, \ldots, X_d=x_d\big] -  \mathrm{VaR}_{0.995}\big[C \big].
\end{align*}
This value change shall  be explained from  the $d$ risk sources. According to the factorization lemma, $F$ is Borel measurable, since  $\mathrm{VaR}_{0.995}[C| X] -  \mathrm{VaR}_{0.995}[C ]$ is a $\sigma(X)$-measurable random variable. So our axioms (A1) to (A9) uniquely define a decomposition of $F$, see Theorem \ref{GenStaticDecompos}.
\end{example}




\section*{References}

\bigskip {\small
\begin{list}{}{\leftmargin1cm\itemindent-1cm\itemsep0cm}
%



\item{Candland, A., Lotz, C., 2014. Profit and Loss attribution.
		In: Internal Models and Solvency II -- From Regulation to Implementation, Risk Books, London.}

\item{Fortin, N., Lemieux, T., Firpo, S., 2011. Decomposition methods in economics.
		 In: Handbook of labor economics. Elsevier, p.~1-102.}


\item{Godin, F., Hamel, E., Gaillardetz, P., Hon-Man Ng, E., 2022. Risk allocation through Shapley decompositions with applications to variable annuities  Available at SSRN: http://dx.doi.org/10.2139/ssrn.4192115.}

\item{Guo, Q., Bauer, D., Zanjani, G., 2021. Capital allocation techniques: Review and comparison. Variance 14(2).}

\item{Junike, G., Stier, H., Christiansen, M.C., 2023. Sequential decompositions at their limit.
ArXiv preprint, arXiv:2212.06733v2.}

\item{Merrick, L., Taly, A., 2020. The Explanation Game: Explaining Machine Learning Models Using Shapley Values. In: Holzinger, A., Kieseberg, P., Tjoa, A., Weippl, E.~(eds) Machine Learning and Knowledge Extraction. CD-MAKE 2020. Lecture Notes in Computer Science, vol 12279. Springer, Cham.}

\item{Moehle, N., Boyd, S., Ang, A., 2021.  Portfolio performance attribution via Shapley value. ArXiv preprint, arXiv:2102.05799.}

\item{Shorrocks, A.F., 2013. Decomposition procedures for distributional analysis: a
			unified framework based on the Shapley value. Journal of Economic Inequality 11(1), 99--126.}

\item{Shapley, Lloyd S., 1953. A Value for n-person Games. In Kuhn, H. W.; Tucker, A. W. (eds.). Contributions to the Theory of Games. Annals of Mathematical Studies. Vol. 28. Princeton University Press,  pp.~307–317.}

\end{list}}

\end{document}